\documentclass[12pt, reqno]{amsart}

\usepackage[margin=2cm]{geometry} 
\usepackage{color}
\usepackage{graphicx}

\usepackage{amsmath}
\usepackage{amssymb}

 \theoremstyle{plain}

\newtheorem*{thm}{Theorem}

\usepackage{cite}

\title{A note on the Cobb-Douglas function}
\author{Richard Vale}
\email{rval012@aucklanduni.ac.nz}
\date{\today}
\begin{document}
\begin{abstract}
This note observes that the Cobb-Douglas function is uniquely characterized by the property that, if the labour share of cost for a constant-returns-to-scale firm remains constant when the firm minimizes its cost for any given output level, then the firm's production function must be Cobb-Douglas.
\end{abstract}
\maketitle
\section{Introduction}
\subsection{} The \emph{Cobb-Douglas production function} \cite{CD} is used in economics to describe the output of a firm or economy with inputs labour $L$ and capital $K$. It is defined as
$$Y(K, L) = A K^\alpha L^{1-\alpha}$$
where $A > 0$ and $0 < \alpha < 1.$ This function is so ubiquitous that it is interesting to look for mathematical properties which characterize it uniquely. Derivations of the Cobb-Douglas function have peadogogical value \cite{Koch2013} but usually require further economic assumptions, such as perfect competition in markets, for example \cite[Sections A,B]{CapitalLabour}.
\subsection{} One motivation (see for example \cite[Section 3.1]{gomme}) for using the Cobb-Douglas function in economic models is that the factor shares of output do not depend on the level of output, in the sense that if $w$ is the wage per unit of labour and $r$ is the rental rate per unit of capital and the economic profit
$$\Pi(K,L) := Y(K,L) - wL - rK$$
attains a global maximum then $wL/Y(K,L) = 1-\alpha$ is constant independent of $Y$. (This fact is sometimes called Bowley's Law \cite[Section 2]{SmirnovWang}.) However, this property does not characterize the Cobb-Douglas function because in general $\Pi(K,L)$ need not have a global maximum.
\subsection{} The aim of the present note is to prove the following theorem, which characterizes the Cobb-Douglas function.
\begin{thm}\label{maintheorem}
Let $Y: \mathbb{R}_{\ge 0}^2 \rightarrow \mathbb{R}_{\ge 0}$ be a differentiable function. Then there are $A > 0, 0 < \alpha < 1$ such that $Y(K,L) = AK^\alpha L^{1-\alpha}$ if and only if the following two conditions hold.
\begin{enumerate}
\item $Y$ has constant returns to scale (i.e. $Y$ is homogeneous of degree $1$).\label{one}
\item There exists $\beta > 0$ such that for all $w, r > 0$ and all $K_0, L_0 > 0$ the following are equivalent.\label{two}
\begin{enumerate}
\item There exists $q>0$ such that the cost function $c(K, L) = rK + wL$ has a unique minimum along the level set (``isoquant") $Y(K,L) = q$ at $(K_0, L_0)$.
\label{conda}
\item $K_0 = L_0 \frac{w}{r} \beta$.\label{condb}
\end{enumerate}
\end{enumerate}
\end{thm}
\begin{proof}
Suppose $Y$ is Cobb-Douglas. Then $Y$ has constant returns to scale. Define $\beta = \frac{\alpha}{1-\alpha}$. Let us show that $(\ref{conda}) \implies (\ref{condb}).$ If $w, r, q > 0$ and $Y(K,L) = q$ then
$c(K, L) = wL + r(q/A)^{1/\alpha}L^{\frac{\alpha -1}{\alpha}}$
and calculus shows that $c(K,L)$ has a unique minimum given by
$$L_0 = \left(\frac{r}{w}\right)^\alpha\left(\frac{q}{A}\right)\beta^\alpha, \qquad K_0 = \left(\frac{r}{w}\right)^{\alpha-1}\left(\frac{q}{A}\right)\beta^{\alpha-1} = L_0 \frac{w}{r} \beta$$ 
as required. To show that $(\ref{condb}) \implies (\ref{conda}),$ let $K_0, L_0$ be given with $K_0 = L_0\frac{w}{r}\beta$ and define $q = AL_0\left(\frac{w}{r}\beta\right)^\alpha$. Then the same algebra shows that $c(K,L)$ has a unique minimum on $Y(K,L)=q$ given by $(K,L) = (K_0, L_0)$. To prove the converse, suppose a function $Y$ satisfies (\ref{one}) and (\ref{two}). Let $K_0, L_0 > 0$ and choose $w$ and $r$ such that $K_0 = L_0\frac{w}{r}\beta$ (for example, let $w=K_0/L_0$ and $r = \beta$). Let $q > 0$ be given by (\ref{conda}). Then by (\ref{conda}) there is a Lagrange multiplier $\lambda$ (which may depend on $w, r, q, K_0$ and $L_0$) satisfying
$w = \lambda\frac{\partial Y}{\partial L}\rvert_{(K_0, L_0)}, r = \lambda\frac{\partial Y}{\partial K}\rvert_{(K_0, L_0)}$
from which it follows that
$$L_0 \frac{\partial Y}{\partial L}\Biggm\rvert_{(K_0, L_0)} \Bigg/ 
K_0 \frac{\partial Y}{\partial K}\Biggm\rvert_{(K_0, L_0)} = \frac{1}{\beta}.$$
Since this holds for any $K_0$ and $L_0$, we get
$$L \frac{\partial Y}{\partial L} =  \frac{1}{\beta} K \frac{\partial Y}{\partial K}$$
but from (\ref{one}), Euler's identity gives $Y = L \frac{\partial Y}{\partial L} + K \frac{\partial Y}{\partial K}$ and so
$$\frac{\partial Y}{\partial L} = \frac{1}{\beta + 1}\frac{Y}{L}, \qquad \frac{\partial Y}{\partial K} = \frac{\beta}{\beta + 1}\frac{Y}{K}$$
from which it follows that $Y = \mathrm{const.} \times K^\alpha L^{1-\alpha}$ with $\alpha = \beta/(\beta + 1)$ as required.
\end{proof}
\section{Remarks}
\subsection{} In words, Theorem \ref{maintheorem} shows the following: any firm with constant returns to scale which minimizes costs of production along each isoquant in such a way that the labour share of total cost $wL/(wL + rK)$ is a constant which is independent of $w, r,$ and the output level $q$ must have a Cobb-Douglas production function. In other words, if the labour share of cost remains constant when cost is minimized for any given output level then the production function is Cobb-Douglas. 
\subsection{} This is a stronger property than the property that the labour share of output is constant when profit is (globally) maximized given $w$ and $r$ because, if the profit function for a firm with constant returns to scale has a global maximum then the maximum profit is zero and so, when profit is maximized, the output $Y(K,L)$ is equal to the cost $wL+rK$.
\subsection{} The author thanks Dr. Q. Gashi for valuable comments.

\bibliography{econbib}
\bibliographystyle{plain}
\end{document}